\RequirePackage{fix-cm}
\documentclass[smallextended]{svjour3}       % onecolumn (second format)
\smartqed  % flush right qed marks, e.g. at end of proof
\usepackage{appendix}
\usepackage{cite}
\usepackage{graphicx}
\usepackage{braket}
\usepackage{amsmath,amssymb}
\usepackage{amsfonts}
\usepackage{graphicx}
\usepackage{comment}
\usepackage{graphicx}
\usepackage{epsfig}
\usepackage{dcolumn}
\usepackage[up]{subfigure}
\usepackage{adjustbox}
\usepackage{float}
\usepackage{dsfont}	
\usepackage{relsize}
\usepackage{verbatim}
\usepackage{float}
\usepackage{ragged2e}
\usepackage{braket}
\usepackage{color}
\usepackage{xcolor}
\usepackage{multirow}
\usepackage{hyperref}
\usepackage[ruled,vlined]{algorithm2e}

\SetCommentSty{mycommfont}

\SetKwInput{KwInput}{Input}                % Set the Input
\SetKwInput{KwOutput}{Output} 
\SetKw{Continue}{continue}
\SetKw{Break}{break}

\smartqed

\begin{document}

\title{On Fault Tolerance of Circuits with Intermediate Qutrit-assisted Gate Decomposition%\thanks{Grants or other notes
}


\author{Ritajit Majumdar \and Amit Saha \and Amlan Chakrabarti \and Susmita Sur-Kolay %etc.
}


\institute{
          Ritajit Majumdar \at
              Advanced Computing \& Microelectronics Unit, Indian Statistical Institute\\
              \email{majumdar.ritajit@gmail.com}
              \and
          Amit Saha \at
              Atos, Pune, India\\
              A. K. Choudhury School of Information Technology, University of Calcutta\\
              \email{abamitsaha@gmail.com}           %  \\
           \and
              Amlan Chakrabarti
          \at   A. K. Choudhury School of Information Technology, University of Calcutta
            \and
              Susmita Sur-Kolay \at 
              Advanced Computing \& Microelectronics Unit, Indian Statistical Institute
}


\maketitle

\begin{abstract}
The use of a few intermediate qutrits for efficient decomposition of 3-qubit unitary gates has been proposed, to obtain an exponential reduction in the depth of the decomposed circuit. An intermediate qutrit implies that a qubit is operated as a qutrit in a particular execution cycle. This method, primarily for the NISQ era, treats a qubit as a qutrit only for the duration when it requires access to the state $\ket{2}$ during the computation. In this article, we study the challenges of including fault-tolerance in such a  decomposition. We first show that any qubit that requires access to the state $\ket{2}$ at any point in the circuit, must be encoded using a qutrit quantum error correcting code (QECC), thus resulting in a circuit with both qubits and qutrits at the outset. Since qutrits are noisier than qubits, the former is expected to require higher levels of concatenation to achieve a particular accuracy than that for qubit-only decomposition. Next, we derive analytically (i)  the number of levels of concatenation required for qubit-qutrit and qubit-only decompositions as a function of the probability of error, and (ii) the criterion for which qubit-qutrit decomposition leads to a lower gate count than qubit-only decomposition. We present numerical results for these two types of decomposition and obtain the situation where qubit-qutrit decomposition excels for the example circuit of the quantum adder by considering different values for quantum hardware-noise and non-transversal implementation of the 2-controlled ternary CNOT gate.

%\subclass{MSC code1 \and MSC code2 \and more}
\end{abstract}

\section{Introduction}

Current quantum devices are engineered to execute a finite number of one or two-qubit gates, termed as the basis gate set \cite{ibmquantum}. Any arbitrary quantum operator $\mathcal{O}$ is equivalent to a cascade of gates taken from this basis gate set \cite{nielsen2002quantum}. This method is often termed as the decomposition of the operator $\mathcal{O}$. A Toffoli gate is a gate with 3 qubits and finds applications in many important applications of quantum computing such as quantum error correction \cite{nielsen2002quantum}, Grover's algorithm \cite{grover}, etc. Several works have focused on the efficient decomposition of the Toffoli gate due to its significance in quantum computing \cite{Selinger_2013, amy, PhysRevA.87.022328}. A trade-off between depth and qubit-count of Toffoli decomposition has been observed in the literature \cite{amy}: the depth of the decomposed circuit can be reduced by allowing ancilla qubits whereas the depth increases when the qubit-count is kept fixed to the original three qubits.

In \cite{gokhalefirst}, the authors proposed temporary usage of $\ket{2}$, a higher dimension state within a qubit system, and showed an exponential reduction in the depth of the decomposition circuit for a Toffoli gate without any ancilla qubits. This has triggered studies on potential applications of temporary usage of $\ket{2}$ within a qubit system. It has been shown to improve the implementation of arithmetic circuits \cite{saha2022intermediate}, and even eliminate the need for SWAP gates in a limited connectivity quantum hardware \cite{saha2022moving}. This approach was generalized to $d \geq 2$ dimensional quantum circuit for multi-controlled Toffoli gates in \cite{saha2022asymptotically}.

Accessing higher dimensions and applying higher dimensional gates nevertheless invoke more errors in the system than qubits. On the other hand, the exponential reduction in the depth of the circuit lowers the effect of amplitude damping error. Numerical studies \cite{gokhalefirst, saha2022asymptotically, saha2022intermediate} show that the overall error on the system is reduced by this method. In other words, the exponential reduction in depth overshadows the temporary usage of higher dimensions. This method of lowering the depth of the circuit via intermediate qutrits is proposed primarily for the Noisy Intermediate Scale Quantum (NISQ) era \cite{Preskill_2018}, where the number of qubits is up to a few hundred, and the qubits are noisy due to the absence of error correction. However, since this method reduces the depth of the decomposed circuit, it is expected to be impactful even in the fault-tolerant era, where we expect thousands of error-corrected qubits in a quantum device.

Adopting this technique to the fault-tolerant era is not without challenges. A fault-tolerant circuit requires the qubits to be encoded at the very beginning. We show in this article that an encoded qubit cannot be made to access dimension $\ket{2}$ temporarily without entertaining the possibility of errors that remain undetected. This mandates that a qubit, which may access dimension $\ket{2}$ temporarily, is encoded and treated as a qutrit throughout the quantum circuit. In other words, the circuit now becomes a hybrid qubit-qutrit circuit. Since qutrits are, in general, noisier than qubits \cite{fischer2022ancilla}, the number of levels of concatenation is expected to be higher to reach a desired accuracy, which, in turn, increases the resource requirement of the circuit. In this paper, we study the trade-off between the reduced circuit complexity that qutrit-based decomposition allows and the increased circuit complexity that this qutrit-based method leads to in the fault-tolerant domain due to a higher level of concatenation. Here, we first show that the level of concatenation for the circuit is determined by the noise profile of a qutrit. We derive an analytical relation between $k_2$ and $k_3$, which are the levels of concatenation required for qubit-only and qubit-qutrit decomposition respectively to achieve the same accuracy. Finally, we analytically derive an inequality that dictates whether the qubit-qutrit decomposition can lead to a fault-tolerant circuit with lower resources for a given size of the circuit and the probability of error. 
%We numerically show this inequality, and thus comment of the applicability on qubit-qutrit decomposition, for the circuits of quantum adder and Grover's algorithm.

The rest of the paper is organized as follows: in Sec.~\ref{sec:background} we discuss some preliminaries on QECC, fault-tolerance, and qubit-qutrit decomposition. Sec.~\ref{sec:ft_dec} dives into deriving the relation between $k_2$ and $k_3$ for a given probability of error. In Sec.~\ref{sec:resource} we derive the criteria for which qubit-qutrit decomposition leads to a circuit with lower resource, and talk about the possibility of transversal implementation of the gates in binary and ternary Shor and Steane code settings. We shed some light on the possible future research in Sec.~\ref{sec:conclusion}.

\section{Background}
\label{sec:background}
We provide a brief background on quantum error correction, fault-tolerant quantum computation, and qutrit-assisted circuit decomposition in this section.
\subsection{Quantum error correction}
\label{sec:qecc}

In quantum error correction, the information of $k$ qubits is encoded into $n>k$ qubits. A unitary error $\mathcal{E}$ in the qubit system can be expressed as a linear combination of the four Pauli operators $\mathbb{I}, \mathbb{X}, \mathbb{Y}, \mathbb{Z}$ \cite{PhysRevA.52.R2493}
\begin{center}
    $\mathcal{E} = a_1\mathbb{I} + a_2\mathbb{X} + a_3\mathbb{Y} + a_4\mathbb{Z}$
\end{center}
where $a_1$, $a_2$, $a_3$, $a_4$ are scalars. Therefore, any quantum error correcting code (QECC), which can correct any of the Pauli errors, can correct any arbitrary unitary error acting on the qubit \cite{PhysRevA.52.R2493}. A distance $d$ QECC can correct upto $t = \lfloor \frac{d}{2} \rfloor$ errors. A distance $d-$ QECC which encodes $k$ qubits into $n$ qubits is represented as $[[n,k,d]]$ QECC.

An $[[n,k,d]]$ QECC fails if more than $t = \lfloor \frac{d}{2} \rfloor$ errors occur. One solution to this is to use a QECC with a larger distance $d$. However, an arbitrarily long quantum computation requires the ability to correct arbitrarily many errors, where the QECC circuit components themselves are noisy as well. This can be countered via fault tolerance \cite{shor1996fault}.

\subsection{Fault tolerance}
Consider a qubit $q$ such that the probability of failure of each component (e.g., state preparation, gate error, measurement error) is upper bound by $p$. If we use a QECC $\mathcal{C}$ to protect each qubit of the circuit then the probability of error drops to $cp^2$, where $c$ denotes the number of ways two errors can occur. For example, if the 9-qubit QECC by Shor \cite{PhysRevA.52.R2493} is used for error correction, then the probability of failure is
$$1 - [(1-p)^9 + \begin{pmatrix}
9\\
1
\end{pmatrix}p(1-p)^8] \approx 36p^2$$
thus $c = 36$. The logical qubit prepared using the QECC can be further protected via concatenation \cite{nielsen2002quantum}. Here, each logical qubit is again protected using $b$ logical qubits, where $b$ is the number of qubits required for encoding ($b=9$ for Shor Code). Then, by the Threshold Theorem \cite{nielsen2002quantum} after $k_2$ levels of concatenation, the probability of failure drops to $\frac{1}{c}(cp)^{2^{k_2}}$, which quickly goes down to $0$ with increasing levels of concatenation as long as $p < \frac{1}{c}$, called the \emph{threshold} value. If the desired probability of error (also termed as \emph{accuracy}) is $\epsilon$, then the number of levels of concatenations required is
\begin{eqnarray}
\frac{1}{c}(cp)^{2^{k_2}} &=& \epsilon \nonumber \\
\Rightarrow k_2 &=& log_2(log_{cp}(c\epsilon)) \nonumber
\end{eqnarray}

Note that the value of $c$ derived above for Shor Code does not take into account measurement error or the error flow from ancilla qubits via multi-qubit operations. Finding the threshold for a given QECC is, in general, non-trivial, and often numerical methods are applied for the same \cite{ni2020neural, varsamopoulos2019comparing, krastanov2017deep, bhoumik2022ml, debasmita22efficient}. Among the QECCs so far, surface code is shown to have the highest threshold of $1\%$ \cite{fowler2009high}, whereas the probability of error of current quantum devices is not below this threshold value \cite{ibmquantum}. Therefore fault-tolerant is not yet achievable with current quantum devices. However, the noise probability is expected to go below the threshold value in the future when fault-tolerant quantum computers are achievable.

In addition, when referring to quantum gates in the fault-tolerant setting, the term ``transversality" describes a characteristic of quantum error-correcting codes that are especially connected to how encoded quantum gates interact with the encoded information. For a quantum gate $G$, the corresponding logical gate $G_L$ is said to be transversal for a given QECC if $G_L = G^{\otimes d}$, where $d$ is the distance of the QECC. In other words, the logical gate is implemented simply by repeating the actual gate on all the qubits used for encoding. Transversality is a desirable property since it makes the implementation simple and cost-effective. However, it should be noted that not all gates have a transversal implementation in all QECCs. For example, $CNOT$, and $H$ gates can be implemented transversally in the Steane Code, but not the $T$ gate. For such gates, non-transversal implementations are designed which are much more costly than the transversal implementation. The non-transversal implementation of $T$ gate in Steane Code is shown in Fig.~\ref{fig:T_gate}. It is often a challenging task to find non-transversal implementation of a certain gate for a given QECC since it can also be probabilistic \cite{gonzalesftqem}. 

\begin{figure}[htb]
    \centering    \includegraphics[width=0.45\textwidth]{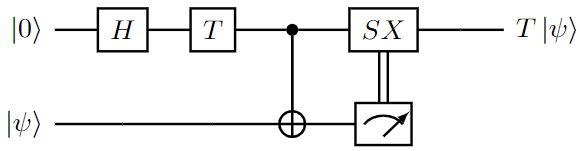}
    \caption{Fault-tolerant implementation of $T$ gate with Steane code}
    \label{fig:T_gate}
\end{figure}

\subsection{Ternary quantum error correction}
Quantum systems are inherently multi-valued. Multi-valued logic provides a larger search space and has been shown to outperform qubit systems in certain domains \cite{saha2018search, saha2021faster, saha2022intermediate}. However, accessing higher dimensions is an engineering challenge, as the noise in the system increases when accessing higher dimensions. Ternary quantum systems, in which an arbitrary qutrit has the mathematical representation of $\alpha\ket{0} + \beta\ket{1} + \gamma\ket{2}$, $\alpha, \beta, \gamma \in \mathbb{C}$, $|\alpha|^2 + |\beta|^2 + |\gamma|^2 = 1$, have been realized experimentally \cite{fischer2022ancilla}. Reliable computing via qutrits requires the implementation of ternary QECCs. In \cite{majumdar2018quantum}, the authors discussed some challenges of the implementation of ternary QECCs, and in \cite{majumdar2022design} the authors showed that any binary QECC can be systematically extended to design its corresponding ternary counterpart. This allows the easy design of ternary QECCs from known binary codes. Furthermore, in this paper, we propose the use of hybrid qubit-qutrit quantum circuits. Such a circuit can be made fault-tolerant by choosing a QECC for the qubits and using its ternary counterpart for the qutrits. This helps to make it an elementary design of fault-tolerant quantum circuits for hybrid qubit-qutrit circuits.

\subsection{Decomposition of gates using higher dimension}

Before discussing the decomposition of the Toffoli gate with the use of higher dimensions, we first discuss the number of gates required for the decomposition of the Toffoli gate in a qubit-only setting. The Toffoli gate can be decomposed in a Clifford + $T$ gate set as per \cite{amy}. A Toffoli gate is decomposed with a T-count of 7 and a Clifford gate-count of 8 as shown in Fig. \ref{fig:tof_selinger}  \cite{amy, Selinger_2013}. 

\begin{figure}[htb]
    \centering
    \includegraphics[width=0.68\textwidth]{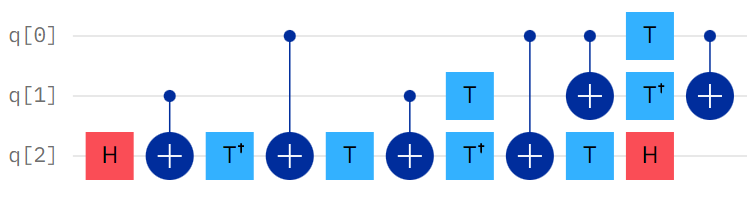}
    \caption{Toffoli decomposition with Clifford+T gates}
    \label{fig:tof_selinger}
\end{figure}

In \cite{gokhale2019asymptotic, 10.1145/3406309}, the authors proposed temporary occupation of the state $\ket{2}$ for decomposition. Maintaining binary input and output allows this circuit construction to be inserted into any pre-existing qubit-only circuits.  A Toffoli decomposition via qutrits has been portrayed in Fig. \ref{tof_qutrit} \cite{gokhale2019asymptotic, 10.1145/3406309}. More specifically, the goal is to carry out a NOT operation on the target qubit (the third qubit $\ket{q_2}$) as long as the two control qubits, are both $\ket{1}$. First, a $\ket{1}$-controlled $X_{+1}$, where $+1$ denotes that the target qubit is incremented by $1 (\text{mod } 3)$, is performed on $\ket{q_0}$ and $\ket{q_1}$, the first and the second qubits. This upgrades $\ket{q_1}$ to $\ket{2}$ if and only if both $\ket{q_0}$ and $\ket{q_1}$ are $\ket{1}$. Then, a $\ket{2}$-controlled $X$ gate is applied to the target qubit $\ket{q_2}$. Therefore, $X$ is executed only when both $\ket{q_0}$ and $\ket{q_1}$ are initially $\ket{1}$. These control qubits are reinstated to their original states by a $\ket{1}$-controlled $X_{-1}$ gate, which reverses the effect of the first gate. That the $\ket{2}$ state from ternary quantum systems can be used instead of an ancilla to store temporary information, is the most important aspect of this decomposition. Thus, to decompose the Toffoli gate, three generalized ternary $CNOT$ gates are sufficient with circuit depth 3. In fact, no $T$ gate is required. Experimental implementation of this decomposition method was shown in \cite{galda2021implementing} using Qiskit pulse simulation.

\begin{figure}[htb]
    \centering
    \includegraphics[scale=.5]{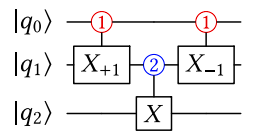}
    \caption{An example of Toffoli decomposition with an intermediate qutrit, where input and output are qubits. The red controls activate on $\ket{1}$
and the blue controls activate on $\ket{2}$. The first gate temporarily elevates $q_1$ to $\ket{2}$ if both $q_0$ and $q_1$ were $\ket{1}$. $X$ operation is then only performed if $q_1$ is $\ket{2}$. The final gate acts as a mirror of the first gate and restores $q_0$ and $q_1$ to their original states \cite{gokhale2019asymptotic}.}
    \label{tof_qutrit}
\end{figure}

This method of Toffoli decomposition is primarily proposed for the NISQ era. However, reduction in depth is always useful, be it the NISQ era or the fault-tolerant era of quantum computation. Therefore, we seek the technicalities required to carry over this method of Toffoli decomposition to the fault-tolerant domain. The primary question is whether the savings in resources (e.g. gate count, lower depth, etc.) are still attainable in a fault-tolerant setting. There are primarily two major concerns for this - (i) ternary systems are more difficult to control, and therefore using the state $\ket{2}$ in a fault-tolerant setting may require significantly more resources due to concatenation, and (ii) whether the 1-controlled and 2-controlled ternary CNOT gates can be implemented transversally in the fault-tolerant setting. In the following part of this paper, we dive deeper into these two queries to understand whether this decomposition of Toffoli gates can be expected to be useful in the fault-tolerant quantum computing era.

\section{Fault-tolerant decomposition using qubit-qutrit gates}
\label{sec:ft_dec}
Consider a quantum circuit $qc$ that has been encoded using a QECC $\mathcal{C}$. This circuit can be decomposed into basis gates using either qubit or qubit-qutrit decomposition method \cite{gokhale2019asymptotic}. Let $p_{2}$ and $p_{2,3}$ be the respective probability of error in these two cases. A few of the qubits behave as qutrits at certain cycles in the qubit-qutrit decomposition \cite{gokhale2019asymptotic}. Literature suggests that these qubits are treated as qutrits only for a limited period of time when it requires access to the state $\ket{2}$. For example, in Fig.~\ref{tof_qutrit}, $\ket{q_1}$ behaves as a qutrit only during the execution of the two qutrit gates on $\ket{q_1}$ and $\ket{q_2}$, and as a qubit otherwise. However,  when we consider error-corrected qubit-qutrit decomposition, we show below in Theorem~\ref{thm:qutrit} that a qubit that requires access to the state $\ket{2}$ at any point in the circuit, must be treated as a qutrit all along. 

\begin{theorem}
\label{thm:qutrit}
Given an error-corrected qubit-qutrit decomposition, a quantum state which requires access to the state $\ket{2}$ at any point in the circuit, must be encoded using a qutrit QECC.
\end{theorem}
\begin{proof}
Without loss of generality, we consider a single-bit flip error only. The proof naturally follows for more complex errors. Let us consider the two operators $X_1$ and $X_2 = X_1^{-1}$ \cite{majumdar2020approximate, majumdar2022design}, where
\begin{center}
    $X_1 = \begin{pmatrix}
    0 & 0 & 1 \\
    1 & 0 & 0 \\
    0 & 1 & 0
    \end{pmatrix}$.
\end{center}

We prove this by contradiction and assume that it suffices to encode all the qubits using a binary QECC. Note that fault tolerance requires all the qubits to be encoded for error correction at the input stage of the circuit \cite{shor1996fault}. Consider a qubit $\ket{\psi} = \alpha\ket{0} + \beta\ket{1}$, which requires access to the computational state $\ket{2}$  at time cycle $t$, has been encoded as $\ket{\psi}_L = \alpha\ket{000} + \beta\ket{111}$ for error correction. When this qubit requires access to state $\ket{2}$, it can be considered as an operator $X_1$ acting on the system, taking $\ket{\psi}_L$ to $\alpha\ket{111} + \beta\ket{222}$. If a bit error occurs on, say, the first qubit while it is accessing the state $\ket{2}$, the erroneous state becomes $\alpha\ket{211} + \beta\ket{022}$. Restoring this system back to a qubit configuration can be considered as applying $X_1^{-1} = X_2$. Applying $X_2$ on this erroneous state takes the state to $\alpha\ket{100} + \beta\ket{211}$, which has undergone leakage from the computational space. Hence, error correction using binary QECC is no longer possible. \qed
\end{proof}

Therefore, if a qubit is allowed to access the state $\ket{2}$ at any point in the circuit, it must be treated as a qutrit from the input stage and encoded accordingly for error correction.

Theorem~\ref{thm:qutrit} does not pose any restrictions on the qubits that do not require access to the state $\ket{2}$ to be encoded using binary QECC for error correction. For example, in Fig.~\ref{tof_qutrit}, $\ket{q_0}$ and $\ket{q_2}$ may be encoded using binary QECC, but $\ket{q_1}$ must be encoded using a ternary QECC.

%transversality and forward pointer to subsection on resource estimation

A general notion is that qutrits and ternary quantum gates are noisier than qubits and binary quantum gates respectively \cite{fischer2022ancilla}. Therefore, it may be possible to attain an accuracy of $\epsilon$ using $k_2$ and $k_3$ levels of concatenation for qubits and qutrits respectively, where $k_2 \leq k_3$. A natural inference is to use fewer levels of concatenation for qubits than for qutrits in the fault-tolerant circuit. In other words, the circuit complexity can be reduced if one can use $k_2< k_3$ levels of concatenation for qubits and qutrits respectively, to acquire the same accuracy for both qubits and qutrits. However, Theorem~\ref{thm:level_concat} below asserts otherwise.

\begin{theorem}
\label{thm:level_concat}
Both the qubits and qutrits in a hybrid qubit-qutrit circuit must be encoded with the same number of levels of concatenation for fault-tolerant implementation irrespective of their respective probability of error.
\end{theorem}

\begin{proof}
Let us assume on the contrary that the qubits and qutrits can be concatenated using $k_2$ and $k_3$ levels of concatenation, where $k_2 \neq k_3$. WLOG, we assume $k_2 < k_3$. Let $G$ be a transversal two-qubit gate for the QECC used, operating over $q_2$ and $q_3$, where $q_2$ $(q_3)$ is a qubit (qutrit). However, since $k_2 < k_3$, the number of qubits encoding $q_2$ is less than the number of qutrits encoding $q_3$. Therefore, by the pigeonhole principle, there exists at least one qubit on which two encoded gates, involving two distinct qutrits, operate. This violates the principle of fault-tolerance since error in this qubit can flow to multiple qutrits. \qed
\end{proof}

Theorem~\ref{thm:level_concat} asserts that the level of concatenation, and hence the size of the circuit, is governed by the probability of error on qutrits.

Thus the fault-tolerant representation of the qubit-qutrit decomposition (i) has a lower depth of the circuit due to its more efficient decomposition, and (ii) leads to circuits with higher resource requirements since it is governed by the probability of errors on qutrits. These two are contrasting characteristics. The aim, therefore, is to determine the criterion for which the increase in the size of the circuit is overshadowed by the number of gates reduced due to more efficient qubit-qutrit decomposition.

In Theorems~\ref{thm:same_k} and ~\ref{thm:same_epsilon} we address: (i) the accuracy of the two types of decomposition if both use the same number of levels of concatenation and (ii) the increase in the number of levels of concatenation required for the qubit-qutrit decomposition to obtain the same accuracy as that of the qubit-only decomposition.

\begin{theorem}
\label{thm:same_k}
Given a quantum circuit $C$, let $C_2$ and $C_{2,3}$ be the two decompositions for $C$ involving qubit gates only and qubit-qutrit gates, with error probabilities $p_2$ and $p_{2,3}$ respectively.  After $k$ levels of concatenation in both, the accuracy $\epsilon_3=\delta \cdot \epsilon_2$ obtained by $C_{2,3}$, where $\epsilon_2$ is the accuracy obtained by $C_2$ and $\delta > 0$, is given by
\begin{equation*}
     log(\delta) = 2^k log(\frac{c_3.p_{2,3}}{c_2.p_2}) + log(\frac{c_2}{c_3})
\end{equation*}
where $\frac{1}{c_2}$ and $\frac{1}{c_3}$ are the thresholds of the binary and the ternary QECCs respectively.
\end{theorem}

\begin{proof}
Follows from the Threshold Theorem \cite{nielsen2002quantum}; see Appendix~\ref{app:same_k}.
\end{proof}

In all the proofs, the logarithm is with respect to base 2 for the sake of ease of calculations and simpler forms of the equations. Note that any other base for the logarithm is equally acceptable for all the calculations. Moreover, in all the theorems henceforth, we have assumed two different thresholds for binary and ternary QECCs. However, if  $c_2 = c_3$, the equations can be modified accordingly to much simpler forms. For example, the equation from this theorem becomes 
\begin{equation*}
    log(\delta) = 2^k log(\frac{p_{2,3}}{p_2})
\end{equation*}

% \begin{corollary}
% The qubit-qutrit decomposition achieves lower error probability than qubit decomposition after $k$ levels of concatenation if $0 < \delta < 1$.
% \end{corollary}

\begin{theorem}
\label{thm:same_epsilon}
Given a quantum circuit $C$, let $C_2$ and $C_{2,3}$ be the two decompositions for $C$ involving qubit gates only and qubit-qutrit gates, with error probabilities $p_2$ and $p_{2,3}$ respectively. Then, $C_2$ as well as $C_{2,3}$ achieves an accuracy of $\epsilon$ after $k_2$ and $k_3$ levels of concatenation respectively, where
\begin{equation*}
    k_3 = \lceil k_2 + log(\frac{log(c_2.p_2) - \frac{1}{2^{k_2}}log(\frac{c_2}{c_3})}{log(\delta) + log(c_3.p_{2})}) \rceil
\end{equation*}

and $\frac{1}{c_2}$,  $\frac{1}{c_3}$ are the thresholds of the binary and ternary QECCs respectively.
\end{theorem}

Before presenting the proof of this theorem, we emphasize here that fault-tolerance is attainable only if $c.p < 1$, where $p$ is the probability of error. Here, we require both $c.p_2$ and $c.p_{2,3} = \delta.c.p_2$ to be less than $1$.

\begin{proof}
Follows from the Threshold Theorem \cite{nielsen2002quantum}; see Appendix~\ref{app:same_epsilon}.
\end{proof}

Currently, surface code is shown to have a threshold of $1\% = 0.01$ \cite{fowler2009high}.  For IBM Quantum devices, $CNOT$ is one of the most prominent sources of error in the quantum systems today, having an error threshold $> 0.01$. Therefore, current quantum devices do not satisfy the fault-tolerance requirement. This prohibits us from providing an experimental comparison of fault-tolerance on the qubit-qutrit decomposition and the qubit-only decomposition. Therefore, for the rest of the paper, any numerical values will assume a futuristic scenario, where $p$ is sufficiently low enough so that $c.p < 1$.

For the sake of better visualization, in Table~\ref{tab:compare}, we assume $c.p_{2, 3} = 0.9$ and $0.5$, and vary $\delta$ to determine the difference in the level of concatenations $k_3$ and $k_2$.
\begin{table}[htb]
    \centering
    \caption{Difference in levels of concatenation $\lceil k_3 - k_2 \rceil$ with varying $\delta$}
    \begin{tabular}{|c|c|c|c|}
        \hline
        $c.p_{2, 3}$ & $\delta$ & $c.p_2$ & $\lceil k_3 - k_2 \rceil$\\
        \hline
        \multirow{6}{*}{0.9} & 1.5 & 0.6 & 3\\
        \cline{2-4}
        & 2 & 0.45 & 3\\
        \cline{2-4}
        & 3 & 0.3 & 4\\
        \cline{2-4}
        & 4 & 0.225 & 4\\
        \cline{2-4}
        & 5 & 0.18 & 5\\
        \hline
        \multirow{5}{*}{0.5} & 1.5 & 0.33 & 1\\
        \cline{2-4}
        & 2 & 0.25 & 1\\
        \cline{2-4}
        & 3 & 0.167 & 2\\
        \cline{2-4}
        & 4 & 0.125 & 2\\
        \cline{2-4}
        & 5 & 0.1 & 2\\
        \hline
    \end{tabular}
    \label{tab:compare}
\end{table}
Since the value of $p_{2,3} \geq p_{2}$, it is expected that $\lceil k_3 - k_2 \rceil \geq 1$. Increasing levels of concatenation increases the size of the resulting circuit exponentially. On the other hand, using qubit-qutrit decomposition is shown to be able to reduce the depth of the circuit exponentially \cite{gokhalefirst}. This implies that the exponential decomposition in the fault-tolerant scenario is useful only if the reduction in depth acquired by such a decomposition overshadows the increase in the circuit size, i.e., the size of the resultant fault-tolerant circuit is not bigger than the qubit-only decomposition. This not only depends on the number of gates in the circuit but also on whether the 1-controlled and 2-controlled ternary CNOT gates can be implemented transversally.

\section{Resource estimation of fault-tolerant circuits}
\label{sec:resource}

Let us assume that the error-correcting code used requires at most $G$ gates to encode a single gate at each level of concatenation in a fault-tolerant manner. Therefore, the size of a circuit consisting of $R$ gates after $k$ levels of concatenation is $G^k \cdot R$. Let the number of gates in qubit-qutrit and qubit-only decomposition of a circuit be $R_{2,3}$ and $R_2$ respectively. Then, we require
\begin{eqnarray}
\label{eq:size}
G^{k_3}\cdot R_{2,3} & \leq & G^{k_2}\cdot R_2 \nonumber \\
\Rightarrow R_{2,3} & \leq & \frac{R_2}{G^{\lceil k_3 - k_2 \rceil}}
\end{eqnarray}

Eq.~(\ref{eq:size}) provides the criteria for which the qubit-qutrit decomposition can result in a smaller circuit even though it requires more levels of concatenation. However, here $G$ is an upper bound on the gate count. This upper bound generally depends on the gates which require non-transversal implementation. Here, instead of comparing the upper bound, we compare the exact gate count for transversal gates, and an estimate of the gate count for non-transversal gates, of certain representative circuits.

Consider a circuit where $\mathcal{G}$ denotes the set of individual gates used in the circuit. Let $n_g$ denote the number of gate type $g \in \mathcal{G}$ in the circuit. Then, after $k$ levels of concatenation, the total number of gates $N_{FT}$ in the fault-tolerant circuit is given in Eq.~(\ref{eq:gate_count}), where $\kappa_g$ denotes the number of gates required for the fault-tolerant implementation of the gate $g$.
\begin{equation}
    \label{eq:gate_count}
    N_{FT} = (\sum_{g \in \mathcal{G}} \kappa_g n_g)^k
\end{equation}

\begin{theorem}
\label{thm:resource}
Let $\mathcal{G}_2$ and $\mathcal{G}_{2,3}$ be the set of types of gates in a circuit, realized respectively by using the qubit-only and the qubit-qutrit decomposition and $n_g$ be the number of gates of type $g$. Then the qubit-qutrit decomposition leads to a smaller number of gates if
%\begin{eqnarray}
\begin{equation}
%\resizebox{0.91\hsize}{!}{%
log(\frac{log(c_2.p_2) - \frac{1}{2^{k_2}}log(\frac{c_2}{c_3})}{log(\delta) + log(c_3.p_{2})}) \leq k_2 \cdot \frac{log\frac{\sum_{g \in \mathcal{G}_2} (\kappa_g n_g)}{\sum_{g \in \mathcal{G}_{2,3}}(\kappa_g n_g)}}{log\sum_{g \in \mathcal{G}_{2,3}}(\kappa_g n_g)}
\label{eq:resource_req}
\end{equation}
where $\frac{1}{c_2}$ and $\frac{1}{c_3}$ are the thresholds of the binary and ternary QECCs respectively.
\end{theorem}

\begin{proof}
See Appendix~\ref{app:resource}.
\end{proof}

Theorem~\ref{thm:resource} provides a relation involving $p_2$ the probability of error, $\frac{1}{c}$ the threshold of the QECC used, $\delta$ the ratio of error probabilities for qutrits and qubits, $k_2$ the levels of the concatenation of qubit-only system, and the gate counts of both types of decompositions.  The value of $\kappa_g$ for each gate $g$ depends on whether it can be implemented transversally or not. For gates that can be implemented transversally, $\kappa_g$ is equal to the distance of the QECC. For other cases, the value of $\kappa_g$ can increase significantly \cite{majumdar2017method} because the implementation may even be probabilistic \cite{gonzalesftqem}.

Next, we consider two well-known codes, namely the Shor and Steane codes, and show that while the 1-controlled ternary CNOT can be implemented transversally in both of them, although in a restricted setting for the Shor code, it is not so the 2-controlled ternary CNOT gate.

\subsection{Transversal implementation of 1-controlled ternary CNOT gate}
In this subsection we discuss the transversal implementation of the 1-controlled ternary CNOT gate for Steane and Shor codes.

\subsubsection{Implementation in Steane code}
Ternary Steane code is a 7-qutrit code whose stabilizers are  \cite{majumdar2023designing}:
\begin{eqnarray}
\label{eq:steane}
S_1 &=& I \otimes I \otimes I \otimes X_1 \otimes X_1 \otimes X_1 \otimes X_1 \nonumber\\ 
S_2 &=& I \otimes X_1 \otimes X_1 \otimes I \otimes I \otimes X_1 \otimes X_1 \nonumber \\
S_3 &=& X_1 \otimes I \otimes X_1 \otimes I \otimes X_1 \otimes I \otimes X_1 \nonumber \\
S_4 &=& I \otimes I \otimes I \otimes Z_1 \otimes Z_2 \otimes Z_2 \otimes Z_1 \nonumber \\
S_5 &=& I \otimes Z_1 \otimes Z_2 \otimes I \otimes I \otimes Z_2 \otimes Z_1 \nonumber \\
S_6 &=& Z_1 \otimes I \otimes Z_2 \otimes I \otimes Z_2 \otimes I \otimes Z_1
\end{eqnarray}

Here, \quad 
%\begin{center}
    $X_1\ket{j} = \ket{j+1} ~\text{mod 3}$, \quad \quad $Z_1\ket{j} = \omega^j \ket{j}$ \quad $Z_2 = Z_1 Z_1$.
%\end{center}
The stabilizers of binary Steane code are similar with only $X$ and $Z$ operators, where \quad 
%\begin{center}
    $X\ket{j} = \ket{j+1} ~\text{mod 3}$, \quad \quad $Z\ket{j} = (-1)^j \ket{j}$.
%\end{center}
For a qubit-qutrit setting with binary and ternary Steane code, it was verified that the 1-controlled ternary CNOT gate can be implemented transversally, i.e., $\text{CNOT}_L(\ket{i}_L \ket{j}_L)$, $i \in \{0,1\}$, $j \in \{0,1,2\}$ can be implemented by executing CNOT gate individually on the 7 qubits used to encode $\ket{i}_L$ and $\ket{j}_L$. Note that in this case, the control is always a qubit, and the target is a qutrit.

\subsubsection{Implementation in Shor code}

% In qubit systems, universal quantum computation requires $T$ gates which is not transversal in many binary QECCs whereas CNOT gates are transversal for binary QECCs like Shor code, Steane code etc. Although the operators in Fig.~\ref{tof_qutrit} are extremely similar to CNOT gates, it is not obvious whether they can be implemented transversally in a qubit-qutrit system. Here we show an example using binary and ternary Shor codes that such an operator is not transversal in this system.}

 For binary systems, let us consider two logical qubits $\ket{1}_L$ and $\ket{0}_L$ encoded using the Shor code, where
\begin{eqnarray*}
    \ket{0}_L &=& \ket{000} + \ket{111} \\
    \ket{1}_L &=& \ket{000} - \ket{111}.
\end{eqnarray*}

We have shown only the first three qubits of the nine qubits encoding and have removed the normalization factor for brevity. An encoded CNOT gate $CNOT_L(\ket{1}_L, \ket{0}_L)$ is shown in Fig.~\ref{fig:encoded_cnot}. Note that the control and target are reversed in the encoded CNOT.

\begin{figure}
    \centering
    \includegraphics[scale=0.5]{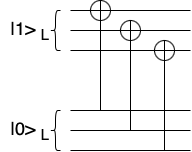}
    \caption{An encoded CNOT gate for a binary quantum system using Shor code}
    \label{fig:encoded_cnot}
\end{figure}

In a similar scenario, let us assume an encoded CNOT gate between two logical qubits $\ket{2}_L$ and $\ket{0}_L$ where the former is encoded using qutrit Shor code and the latter using qubit Shor code. As before, we shall ignore the normalization factor and only report the first three qubits (qutrits) for brevity.
\begin{eqnarray*}
    \ket{2}_L &=& \ket{000} + \omega^2\ket{111} + \omega \ket{111} \\
    \ket{0}_L &=& \ket{000} + \ket{111}
\end{eqnarray*}

The encoded CNOT operation between these two states (where the control is on $\ket{0}_L$ and the target is on $\ket{2}_L$ similar to Fig.~\ref{fig:encoded_cnot}) is
\begin{eqnarray*}
    && CNOT[(\ket{000} + \ket{111})(\ket{000} + \omega^2\ket{111} + \omega \ket{111})] \\
    &=& CNOT[\ket{000}(\ket{000} + \omega^2\ket{111} + \omega \ket{111})] + CNOT[\ket{111}(\ket{000} + \omega^2\ket{111} + \omega \ket{111})] \\
    &=& \ket{000}(\ket{000} + \omega^2\ket{111} + \omega \ket{111}) + \ket{111}(\ket{111} + \omega^2\ket{222} + \omega \ket{000}) \\
    &=& \ket{000}(\ket{000} + \omega^2\ket{111} + \omega \ket{111}) + \omega\ket{111}(\ket{000} + \omega^2\ket{111} + \omega \ket{111}\\
    &=& (\ket{000} + \omega\ket{111})(\ket{000} + \omega^2\ket{111} + \omega \ket{111})
\end{eqnarray*}
We note here that the state of $\ket{0}_L$ has deviated from the codespace and hence will be treated as an error. 

This issue may be tackled by encoding all the states as qutrits from the very beginning, irrespective of whether these require access to state $\ket{2}$. Nevertheless, this deviates from the qubit-qutrit setting which is the primary focus of this article. Therefore, we do not delve deeper along this line here.

\subsubsection{Implementation of 2-controlled ternary CNOT gate}
The 2-controlled ternary CNOT gate does not seem to adhere to any transversal implementation for both Steane and Shor QECCs. For Steane code, once again a brute force approach may be taken to check this. For Shor code, we already established that all the states must be qutrit to ensure transversal implementation of 1-controlled ternary CNOT gate, but the action of a 2-controlled ternary CNOT is $\ket{0} \leftrightarrow \ket{1}$ when the control qutrit is $\ket{2}$. This requirement does not allow a transversal setting since it has to ignore the state $\ket{2}$ in the target. This is not achievable by general ternary NOT gates which operate as addition modulo 3.

Therefore, for both of these QECCs, extra gates and possibly extra qutrits are needed to achieve the required action of a 2-controlled ternary CNOT gate, making its implementation non-transversal. For the following numerical study, we vary the value of $\kappa_g$ for these gates to determine the value for which the fault-tolerant qubit-qutrit decomposition still leads to a lower number of gates as opposed to the qubit-only decomposition.

\subsection{Overview of circuit decomposition for the adder}

As an illustration, we evaluate the inequality of Theorem~\ref{thm:resource} for an adder circuit. In this circuit, we decompose the Toffoli gate as per Fig.~\ref{fig:tof_selinger}. Let $Toffoli\_count$, $CNOT\_count$, $T\_count$, and $H\_count$ denote respectively the total count of Toffoli, $CNOT$, $T$, and Hadamard gates required.

\textbf{Adder}: The circuit for adding in place two $n$-qubit registers, has \cite{draper2006logarithmic}, 
\begin{multline}
 Toffoli\_count_{add} =  10n-3w(n)-3w(n-1)-3\log_2n\\ - 3\log_2(n-1) -7
\end{multline}
where $w(n)$ denotes the number of ones in the binary representation of $n$. For the sake of simplicity, we consider $w(n) = n$ and $w(n-1) = n-1$. With these values, the gate counts for the qubit-only decomposition of the Toffoli gates are given in Table~\ref{tab:adder_qubit}.

\begin{table}[htb]
    \centering
    \caption{Gate counts for qubit-only decomposition of Toffoli gate for adder circuit}
    \begin{tabular}{|c|c|}
    \hline
        Gate & Gate Count  \\
        type & (as a function of $n$)\\
    \hline
        $CNOT$ & $24n - 18\log_2n - 18\log_2(n-1) - 24$\\
        \hline
        $H$ & $8n -6\log_2n - 6\log_2(n-1) -8$\\
        \hline
        $T$ & $14n - 28\log_2n - 28\log_2(n-1) -21$\\
        \hline
    \end{tabular}
    \label{tab:adder_qubit}
\end{table}

By qubit-qutrit decomposition, the Toffoli gate is decomposed using only 1-controlled and 2-controlled ternary CNOT gates. From Fig.~\ref{tof_qutrit} we note that the number of 1-controlled ternary CNOT gates required for a Toffoli decomposition is twice that of 2-controlled ternary CNOT gates. The total number of 1-controlled and 2-controlled ternary CNOT gates required for qubit-qutrit decomposition of the adder circuit, as obtained from \cite{saha2022intermediate}, are:
\begin{eqnarray*}
    \# \text{1-controlled ternary CNOT} &=& 8n -6\log_2n - 6\log_2(n-1) - 8 \\
    \# \text{2-controlled ternary CNOT} &=& 4n-3\log_2n - 3\log_2(n-1) - 4.
\end{eqnarray*}

\subsection{Comparison of resource requirements}
In the previous subsection, it has been shown that qubit-qutrit decomposition is capable of removing the requirement of $T$ gates from the circuit, but the 2-controlled ternary CNOT gate still remains non-transversal like the $T$ gate. Since the $T$ gate is a non-transversal gate for many QECCs, it has a much costlier implementation. Several techniques have been proposed in the literature for efficient fault-tolerant implementation of $T$ gates \cite{postler2022demonstration, litinski2019game}, which primarily aim for surface codes, and often involve complicated processes such as teleportation. For the sake of simplicity,  we consider here the fault-tolerant decomposition of $T$ gates in Steane code \cite{nielsen2002quantum} as shown in Fig.~\ref{fig:T_gate}. The $H$, $T$, and $SX$ gates are implemented transversally to attain the overall effect of an encoded $T$ gate. Therefore, if Steane code is used, then $\kappa_g$ for transversal gates is $7$, whereas that for $T$ gates is $4\times7 = 28$. Note that this implementation requires ancilla qubits. Therefore, if an $n$-qubit circuit involves $m$ $T$ gates, it leads to an overall circuit with $n+m$ qubits.
\begin{comment}
\begin{figure}[htb]
    \centering
    \begin{quantikz}
    \lstick{$\ket{0}$} & \gate{H} & \gate{T} & \ctrl{1} & \gate{SX} & \qw & \rstick{$T\ket{\psi}$} \qw \\
    \lstick{$\ket{\psi}$} & \qw & \qw & \targ{} \vqw{-1} & \meter{} \vcw{-1} \\
    \end{quantikz}
    \caption{Fault tolerant implementation of $T$ gate with Steane code}
    \label{fig:T_gate}
\end{figure}
\end{comment}

Let $N_2$ and $N_{2,3}$ denote the total number of gates required for the fault-tolerant implementation of the adder using qubit-only and qubit-qutrit decomposition for a single level of concatenation. The number of gates increases exponentially with the number of levels of concatenation. We use $N_2(g)$ to denote the number of gates $g$ in the fault-tolerant implementation for qubit-only decomposition and $N_{2,3}(g)$ for qubit-qutrit decomposition.

Then for the qubit-only decomposition of an adder, we have
\begin{eqnarray}
\label{eq:adder_binary}
N_2 &=& 7 \times (N_2(CNOT) + N_2(H) \nonumber + 4 \times N_2(T)) \nonumber \\
&=& 616n - 952 log_2 n - 952 log_2 (n-1) - 798.
\end{eqnarray}

We have also discussed earlier that the qubit-qutrit structure can be retained for Steane Code where the 1-controlled ternary CNOT gate is transversal but the 2-controlled ternary CNOT is not.

Then for the qubit-qutrit decomposition of an adder, we have
\begin{eqnarray}
\label{eq:adder_ternary}
N_{2,3} &=& 7 \times N_{2,3}(\text{1-controlled ternary CNOT}) + N_{2,3}(\text{2-controlled ternary CNOT}) \nonumber \\
&=& (56n - 42 log_2 n - 42 log_2 (n-1) - 56) \nonumber \\
&+&  \kappa_g \times (28n - 21 log_2 n - 21 log_2 (n-1) - 28).
\end{eqnarray}

where $\kappa_g$ denotes the number of gates required for the non-transversal implementation of the 2-controlled ternary CNOT gate in Steane code. If a transversal implementation were possible, then $\kappa_g = 1$. In the following subsection, we shall vary the value of $\kappa_g$ to determine the scenarios (i.e. the gate cost of non-transversal implementation of the 2-controlled ternary CNOT gate) for which qubit-qutrit decomposition leads to a lower gate cost. For simplicity, in the numerical analysis, we assume $c_2 = c_3 = c$.

\subsubsection{Numerical analysis}

In Fig.~\ref{fig:adder_fig} we present the numerical results for an adder with varying number of qubits.
\begin{figure}[htb]
    \centering
    \subfigure[50 qubit adder]{
        \includegraphics[scale=0.22]{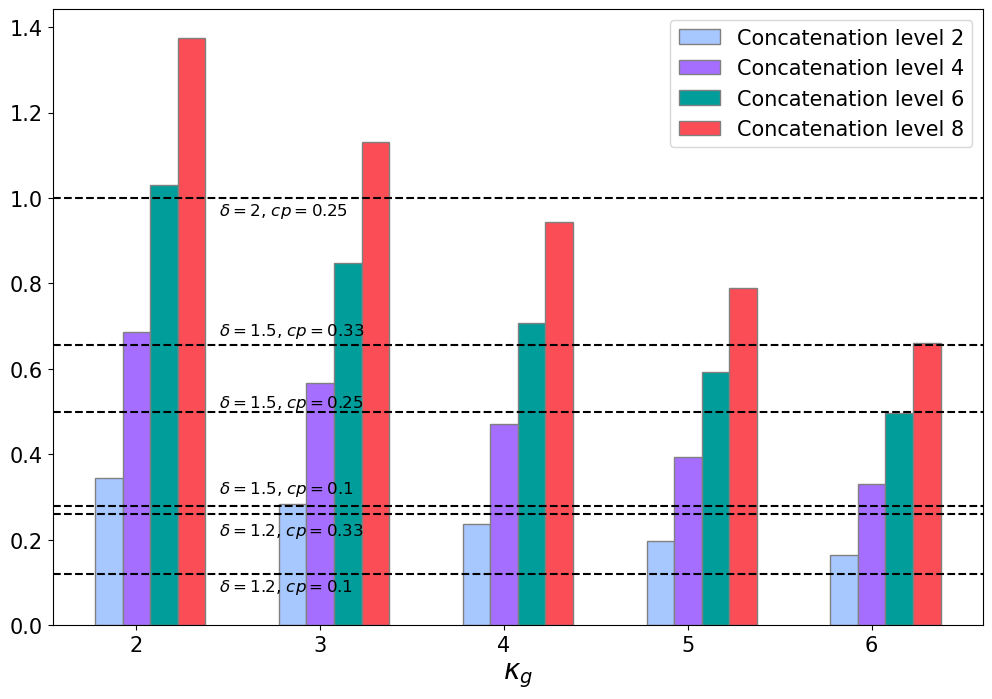}
    }
    \hfill
    \subfigure[100 qubit adder]{
        \includegraphics[scale=0.22]{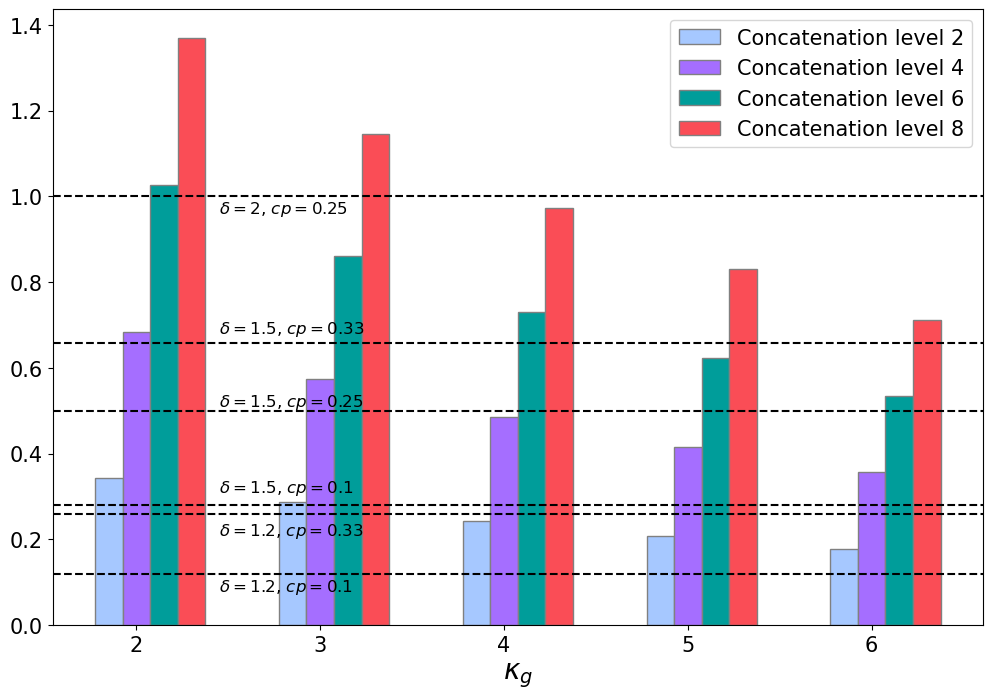}
    }

    \subfigure[300 qubit adder]{
        \includegraphics[scale=0.22]{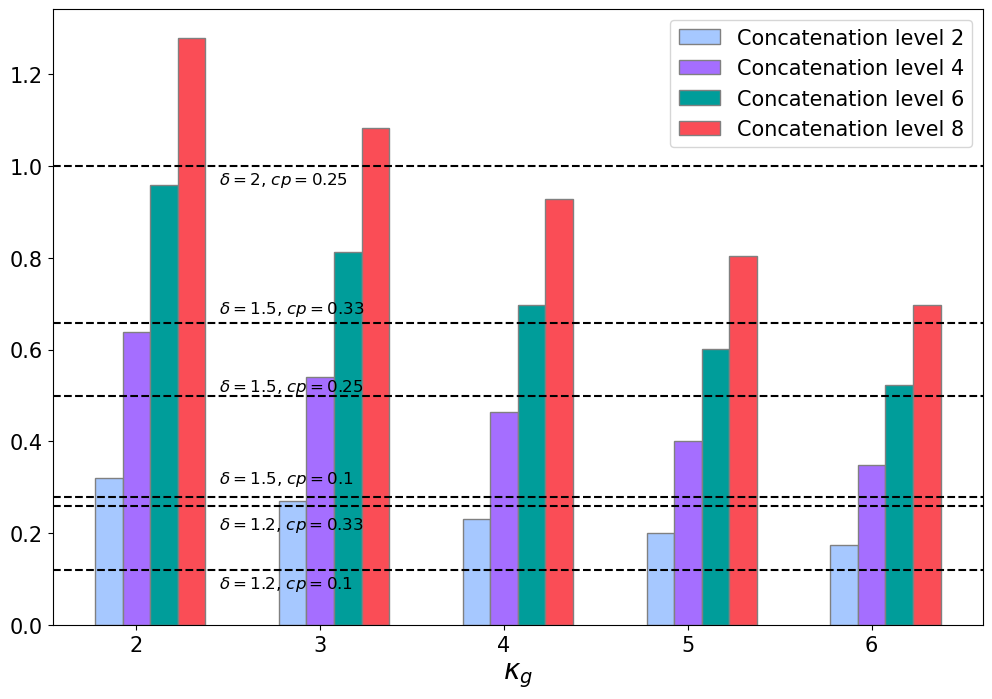}
    }
    \hfill
    \subfigure[800 qubit adder]{
        \includegraphics[scale=0.22]{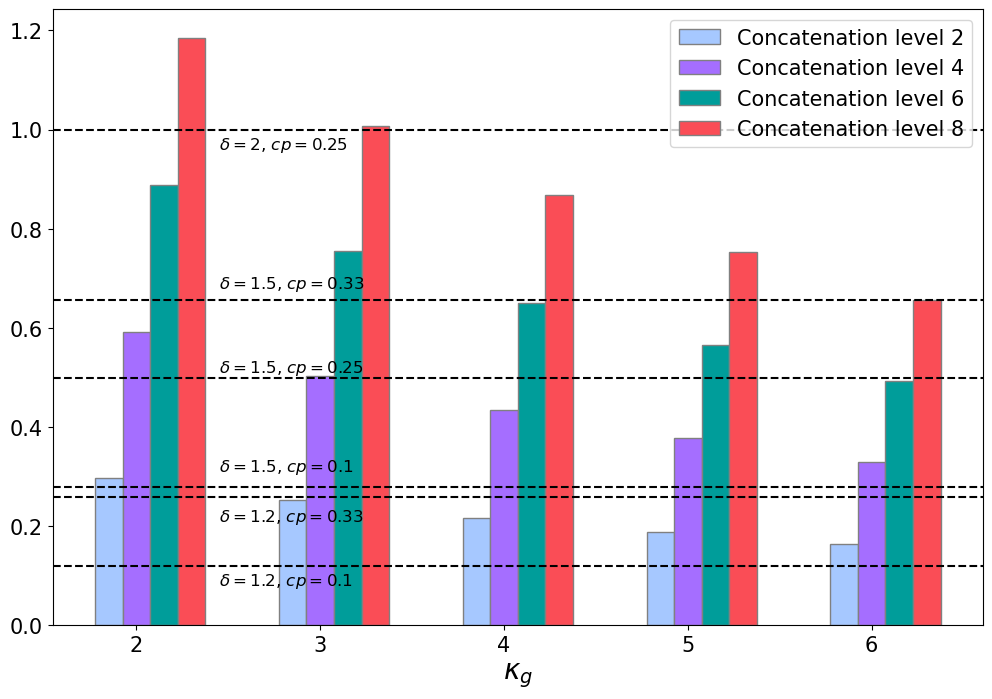}
    }
    \caption{For $2 \leq \kappa_g \leq 6$ (refer to Eq.~(\ref{eq:adder_ternary}),  the values of the RHS of the inequality of Eq.~(\ref{eq:resource_req}) are the heights of the bar-plots. The LHS of the inequality is indicated by the horizontal dashed lines for different values of $\delta$ and $c.p$. The qubit-qutrit decomposition leads to lower resource requirements for certain $\delta$ and $c.p$ when the LHS of Eq.~(\ref{eq:resource_req}) is less than RHS, i.e., the bar plots are higher than the corresponding horizontal dashed line.}
    \label{fig:adder_fig}
\end{figure}
The resource requirement of the qubit-qutrit decomposition is lower than that for the qubit-only decomposition if the LHS of the inequality of Eq.~\ref{eq:resource_req} is less than that of the RHS. In  Fig.~\ref{fig:adder_fig}, the RHS is shown as bars for different values of $\kappa_g$ and levels of concatenation. The LHS for different values of $\delta$ and $c.p$ are shown as horizontal dashed lines. Therefore, the requirement of the inequality translates to the fact that qubit-qutrit decomposition requires a lower number of gates than that for the qubit-only decomposition if the bar is higher than the corresponding horizontal dashed line. The primary observations from the figures are summarized below:

\begin{enumerate}
    \item[(i)] We note that for a fixed value of $\kappa_g$, increasing the level of concatenation increases the height of the bars. This implies that if a higher level of concatenation is used, then the qubit-qutrit decomposition eventually triumphs because of the exponential reduction in the number of gates required for the same. On the other hand, as the value of $\kappa_g$ increases, the heights of the bars lowers. This is also natural since as the cost of the non-transversal implementation of the 2-controlled ternary CNOT gate increases, the benefit of the exponential reduction in gate count by the qubit-qutrit decomposition also lowers.

    \item[(ii)] The value of horizontal plots increases with the value of $\delta$ and $c.p$. In other words, the noisier the qutrit hardware, the more difficult it is to obtain lower resources using qubit-qutrit decomposition. If in the future, the noise profile of qubit and qutrit devices become similar, then we expect that this qubit-qutrit decomposition will lead to lower resources even for a high value of $\kappa_g$.

    \item[(iii)] Finally, as we increase the number of qubits, we note that the height of the bars increases initially and then decreases again. Therefore, apart from the value of $\kappa_g$ and the noise profile of the hardware, the number of qubits also plays a role in determining whether qubit-qutrit decomposition can provide benefit in terms of resource.
\end{enumerate}

This figure illustrates our key observations with an adder circuit. Similar observations may be performed for other circuits of interest that require the decomposition of the  Toffoli gates.

\section{Concluding remarks}
\label{sec:conclusion}

In this paper, we analytically studied the challenges of extending the intermediate qutrit-based decomposition to the fault-tolerant regime. In general, qutrits are noisier than qubits and hence require higher levels of concatenation to attain the same level of accuracy. We derive an analytical formula for which the overall resource requirement of qubit-qutrit decomposition is lower than that of qubit-only decomposition. We provide an inequality for the resource requirement of the qubit-only and qubit-qutrit decomposition which shall govern whether it is useful to apply this method of decomposing Toffoli gates in a fault-tolerant setting.

This study opens up a myriad of research directions. Primarily the study of resources for different binary and ternary QECCs by finding the proper transversal and non-transversal implementation of the gates is a domain of interest. Such a study can provide more insight into the settings where this type of decomposition will be useful in the fault-tolerant era. It is also a domain of interest whether using higher dimensions allows transversal (or non-transversal implementation with a low value of $\kappa_g$) implementation of the 2-controlled ternary CNOT gates. These future prospects are beyond the scope of this article.

%Our numerical results show that for small values of $\delta$, such as $1.5, 2$, qubit-qutrit decomposition is shown to have lower resource when $k_2$ increases. In this study, we have seen the numerical results for $\delta = 1.5, 2$ only. It is evident that as $\delta$ increases, qubit-qutrit decomposition will fail to show improvement. It is of interest to study the value of $\delta$ that can be derived for current quantum computers. Furthermore, our numerical results show that when $\delta = 1.5$ and $c.p_2 = 0.6$, a concatenation level of at least 12 is required to achieve improvement for qubit-qutrit decomposition. Even then, the improvement is lost for $n$ (number of qubits) beyond 400 as shown in Figure \ref{fig:delta_15_high}. It is of interest to check whether a concatenation level of 12 is at all necessary for $n \leq 400$ qubits to achieve an accuracy which is sufficient for practical quantum computation. A follow up study in this direction can further elaborate on possibility of the qubit-qutrit decomposition to provide improvement in the fault-tolerant era.

\begin{acknowledgements}
The authors have no conflicts of interest to declare that are relevant to the content of this article.
\end{acknowledgements}

\textbf{Data Availability:} Our manuscript has no associated data.

%%%%%%%%%%%%%%%%%%%%%%%%%%%%%%%%%%%%%%%%%
\section*{Appendix}

\subsection*{Proof of Theorem~\ref{thm:same_k}}
\label{app:same_k}

\begin{proof}
Let the accuracy obtained using only qubit and qubit-qutrit decomposition after $k$ levels of concatenations be $\epsilon_2$ and $\epsilon_3$ respectively, where $\epsilon_3 = \delta \cdot \epsilon_2$. Therefore, for a QECC with threshold $\frac{1}{c}$,
\begin{eqnarray}
\frac{1}{c}(c.p_{2,3})^{2^k} & = & \frac{\delta}{c}(c.p_2)^{2^k} \nonumber \\
\Rightarrow 2^k \cdot log(c.p_{2,3}) & = & log(\delta) + 2^k \cdot log(c.p_2) \nonumber \\
\Rightarrow 2^k \cdot log(c.p_{2,3}) - 2^k \cdot log(c.p_2) & = & log(\delta) \nonumber \\
\Rightarrow log(\frac{p_{2,3}}{p_2}) & = & \frac{log(\delta)}{2^k} \nonumber \\
\Rightarrow log(\delta) & = & 2^k log(\frac{p_{2,3}}{p_2}) \nonumber
\end{eqnarray} \qed
\end{proof}

\section*{Proof of Theorem~\ref{thm:same_epsilon}}
\label{app:same_epsilon}

\begin{proof}
The accuracy obtained after $k$ levels of concatenation is $\frac{1}{c}(c.p)^{2^k}$, where $p$ is the probability of error. In our setting, both types of decomposition attain the same accuracy after $k_2$ and $k_3$ levels of concatenations. Therefore,
\begin{eqnarray}
\frac{1}{c}(c.p_{2,3})^{2^{k_3}} & = & \frac{1}{c}(c.p_2)^{2^{k_2}} \nonumber \\
\Rightarrow 2^{k_3 - k_2} log(c.p_{2, 3}) &=& log(c.p_2) \nonumber
\end{eqnarray}
Since $c.p_{2, 3} < 1$ as per the requirement of fault-tolerance, $log(c.p_{2, 3}) < 0$. Dividing both sides by $log(c.p_{2, 3})$ we get
\begin{eqnarray}
2^{k_3 - k_2} & = & \frac{log(c.p_2)}{log(c.p_{2, 3})} \nonumber \\
\Rightarrow k_3 - k_2 &=& log(\frac{log(c.p_2)}{log(c.p_{2, 3})}) \nonumber \\
%\Rightarrow k_3 &=& k_2 + log(\frac{log(c.p_2)}{log(c.p_{2, 3})}) \nonumber \\
\Rightarrow k_3 &=& k_2 + log(\frac{log(c.p_2)}{log(c.\delta.p_2)}) \nonumber \\
\Rightarrow k_3 &=& k_2 + log(\frac{log(c.p_2)}{log(\delta) + log(c.p_2)}) \nonumber
\end{eqnarray} \qed
\end{proof}

\section*{Proof of Theorem~\ref{thm:resource}}
\label{app:resource}
\begin{proof}
We require the overall gate count of the qubit-qutrit decomposition to be lower than that of the qubit decomposition. In other words, our requirement is
\begin{equation*}
    \label{eq:gate_comp}
    (\sum_{g \in \mathcal{G}_{2,3}} \kappa_g n_g)^{k_3} \leq (\sum_{g \in \mathcal{G}_{2}} \kappa_g n_g)^{k_2}
\end{equation*}

Now,
\begin{eqnarray}
\label{eq:gate_comp_eq}
(\sum_{g \in \mathcal{G}_{2,3}} \kappa_g n_g)^{k_3} &\leq& (\sum_{g \in \mathcal{G}_{2}} \kappa_g n_g)^{k_2} \nonumber \\
\Rightarrow k_3 log(\sum_{g \in \mathcal{G}_{2,3}} \kappa_g n_g) &\leq& k_2 log(\sum_{g \in \mathcal{G}_{2}} \kappa_g n_g) \nonumber \\
\Rightarrow \frac{k_3}{k_2} &\leq& \frac{log(\sum_{g \in \mathcal{G}_{2}} \kappa_g n_g)}{log(\sum_{g \in \mathcal{G}_{2,3}}\kappa_g n_g)} \nonumber \\
\Rightarrow \frac{k_3}{k_2} - 1 &\leq& \frac{log(\sum_{g \in \mathcal{G}_{2}} \kappa_g n_g)}{log(\sum_{g \in \mathcal{G}_{2,3}}\kappa_g n_g)} - 1 \nonumber
\end{eqnarray}
\begin{equation}
\Rightarrow k_3 - k_2 \leq k_2 \cdot \frac{log(\frac{\sum_{g \in \mathcal{G}_{2}} \kappa_g n_g}{\sum_{g \in \mathcal{G}_{2,3}}\kappa_g n_g})}{log(\sum_{g \in \mathcal{G}_{2,3}}\kappa_g n_g)}
\end{equation}

By Theorem~\ref{thm:same_epsilon}, we have\\ $k_3 = k_2 + log(\frac{log(c.p_2)}{log(\delta) + log(c.p_2)})$. Substituting this in Eq.~(\ref{eq:gate_comp_eq}), we have

\begin{equation*}
log\frac{log(c.p_2)}{log\delta + log(c.p_2)} \leq k_2 \cdot \frac{log\frac{\sum_{g \in \mathcal{G}_{2}} (\kappa_g n_g)}{\sum_{g \in \mathcal{G}_{2,3}}(\kappa_g n_g)}}{log\sum_{g \in \mathcal{G}_{2,3}}(\kappa_g n_g)}
%\label{eq:resource_req1}
\end{equation*} \qed
\end{proof}

\bibliographystyle{spmpsci} 
%\bibliographystyle{qinp}
%\bibliography{mybib}
\bibliography{Ref}

\end{document}